\documentclass[a4paper,preprint,pra,twocolumn,10pt,twoside,showpacs]{revtex4}%
\usepackage{geometry}
\usepackage{amsfonts}
\usepackage{amsmath}
\usepackage{amssymb}
\usepackage{hyperref}
\usepackage{graphicx}
\usepackage{color}%
\usepackage{amsbsy}
\usepackage{bbm}
\usepackage{bm}
\newcommand{\id}{\mathbbm{1}}
\providecommand{\U}[1]{\protect\rule{.1in}{.1in}}
\newtheorem{theorem}{Theorem}

\newenvironment{proof}[1][Proof]{\noindent\textbf{#1.} }{\ \rule{0.5em}{0.5em}}
\begin{document}
\preprint{ }
\title[Short title for running header]{Determining lower bounds on a measure of multipartite entanglement from few local observables}
\author{Jun-Yi Wu, Hermann Kampermann, Dagmar Bru{\ss}}
\affiliation{Institut f\"ur Theoretische Physik III, Heinrich-Heine-Universit\"at D\"usseldorf, D-40225 D\"usseldorf, Germany}
\author{Claude Kl\"ockl}
\affiliation{University of Vienna, Faculty of Mathematics, Nordbergstraße 15, 1090 Wien, Austria}
\author{Marcus Huber}
\affiliation{University of Bristol, Department of Mathematics, Bristol, BS8 1TW, U.K.}

\begin{abstract}
We introduce a method to lower bound an entropy-based measure of genuine multipartite entanglement via nonlinear entanglement witnesses. We show that some of these bounds are tight and explicitly work out their connection to a framework of nonlinear witnesses that were published recently. Furthermore, we provide a detailed analysis of these lower bounds in the context of other possible bounds and measures. In exemplary cases, we show that only a few local measurements are necessary to determine these lower bounds.
\end{abstract}
\volumeyear{year}
\volumenumber{number}
\issuenumber{number}
\eid{identifier}
\date[Date text]{date}
\received[Received text]{date}

\revised[Revised text]{date}

\accepted[Accepted text]{date}

\published[Published text]{date}

\maketitle

%\tableofcontents

\section{Introduction}

Quantum entanglement is central to the field
of quantum information theory. Due to its numerous applications in upcoming
quantum technology much research has been devoted to its understanding
 (for a recent overview consider Ref.~\cite{horodeckiqe}).\newline
Especially in systems comprised of many particles entanglement provides
numerous challenges and of course potential applications, such as building
quantum computers (see Ref.~\cite{qc}), performing quantum algorithms (the
connection to multipartite entanglement is demonstrated in Ref.~\cite{qa}) and
multi-party cryptography (see e.g. Ref.~\cite{SHH3}).\newline Furthermore, the
understanding of the behavior of complex systems seems to be closely linked to
the understanding of multipartite entanglement manifestations, demonstrated by
the connection to phase transitions and ionization in condensed matter systems
(e.g. \cite{cond}), the properties of ground states in relation to
entanglement (as shown e.g. in Ref.~\cite{spin,ground}), or potentially even
biological systems (such as e.g. bird navigation \cite{bird}).\newline In
order to judge the relevance of entanglement in such systems it is crucial to
not only detect its presence, but also quantify the amount. The structure of entangled states, especially in multipartite systems \cite{acin}, is very complex and
the question whether a given state is entangled is even NP-hard \cite{gurvits}. Thus, in general, it will
not be possible to derive a computable measure of entanglement that reveals all entangled states to be entangled and discriminates between different entanglement classes. Furthermore, %in large and complex systems
full information about the state of the system requires a number of measurements that grows exponentially in the
size of the system. For the detection of
entanglement in multipartite systems most researchers have therefore made it a
primary goal to develop entanglement witnesses, which via a limited amount of
local measurements can detect the presence of entanglement, even in complex
systems (for an overview of multipartite entanglement witnesses consider
Ref.~\cite{guehnetoth}).\newline The expectation value  of witness-operators
are usually expressed in terms of inequalities, which if violated show the presence of entanglement.
Nonlinear witnesses (first introduced in Ref.~\cite{horodeckinonlinear} see also early discussions in e.g. Ref.~\cite{nonlin}) provide a generalization that is no longer a linear function
of density matrix elements, but a nonlinear one. Thus one cannot reformulate the criteria in terms of an expectation value of a hermitian operator (unless one considers coherent measurements on multiple copies of the state, which out of experimental infeasibility we do not discuss in our manuscript). We will henceforth refer to inequalities that involve nonlinear functions of density matrix elements as nonlinear entanglement witnesses.\\
Recently some authors pointed out a connection between the possible amount of violation of these nonlinear inequalities
and quantification of entanglement in multipartite systems (in
Ref.~\cite{Guehnetaming} and Ref.~\cite{maetal}).\newline The aims of this
paper are twofold. First to systematically show the connection of numerous witnesses
to a meaningful measure of genuine multipartite entanglement and second to use this
established relation for the development of novel witnesses, which by
construction give lower bounds on that measure. To that end we follow and
generalize the approach from Ref.~\cite{maetal}.\\
It turns out that only a small number of density matrix elements enters into our lower bounds, making the construction experimentally feasible even in larger systems of high dimensionality.

\section{A measure of multipartite entanglement and its lower bounds}

%\subsection{Genuine Multipartite Entanglement (GME) Measure $E_{m}%
%$\label{Sec. Note & Def. GME Measure}}

\subsection{A measure of genuine multipartite entanglement (GME)}

%A genuine multipartite entanglement measure (GME measure) $E_{m}%
%$\cite{Ma&Chen&Chen@2011-GME_MEASURE} is introduced by Z.Ma, Z.Chen, J.Chen
%and etc., which is constructed via a convex roof over all possible
%decomposition%
The entropy of subsystems has often been used, in order to quantify
entanglement contained in multipartite pure states (e.g. see
\cite{horodeckiqe,Milburn,Love,HH2,HHK1}). In this paper we will follow the
definition first presented in Ref.\cite{Milburn} and define a measure of GME
for multipartite pure states as
\begin{align}
E_{m}(|\psi\rangle\langle\psi|):=\min_{\gamma}\sqrt{S_{L}\left(  \rho
_{\gamma}\right)  }=\min_{\gamma}\sqrt{2\left(  1-\text{Tr}(\rho_{\gamma}%
^{2})\right)  }\,, \label{Def. GME measure}%
\end{align}
where
$S_{L}\left(  \rho_{\gamma}\right)  $ is the linear entropy of the
reduced density matrix of subsystem $\gamma$, i.e. $\rho_{\gamma}%
:=\text{Tr}_{\bar{\gamma}}(|\psi\rangle\langle\psi|)$. The minimum is taken over all possible reductions $\gamma$ (where the complement is denoted as $\bar{\gamma}$), which corresponds to a bipartite split into $\gamma|\bar{\gamma}$.
\newline As any proper
measure of multipartite entanglement for pure states can be generalized to
mixed states via a convex roof, i.e.
\[\label{definitionroof}
E_{m}(\rho):=\inf_{\{p_{i},|\psi_{i}\rangle\}}\sum_{i}p_{i}E_{m}(|\psi
_{i}\rangle\langle\psi_{i}|)\,.
\]
Due to its construction this measure fulfills almost all desirable properties
one would expect from measures of GME (see Ref.~\cite{maetal} for details).
Because computing all possible pure state decompositions of a density matrix is
computationally impossible even if one is given the complete density matrix,
we require lower bounds to be calculable for this expression.\\
Also note that a lower bound on the linear entropy directly leads to a lower bound on the R{\'e}nyi $2$-entropy $S_R^{(2)}(\rho_\gamma)$ via the relation
$S_R^{(2)}(\rho_\gamma)=-\log_2(\frac{2-S_L(\rho_\gamma)}{2})$, which also provides one of the physical interpretations
of this measure. The R{\'e}nyi $2$-entropy in itself is a lower bound to the von Neumann entropy $S(\rho_\gamma)$ and the mutual information can be expressed as $I_{\gamma\bar{\gamma}}:=S(\rho_\gamma)+S(\rho_{\bar{\gamma}})-S(\rho)=2S(\rho_\gamma)$. Thus by our lower bound we gain a lower bound on the average minimal mutual information across all bipartitions of the pure states in the decomposition, minimized over all decompositions.

\subsection{Linear entropy and its convex roof}

The state vector of an $n$-partite qudit state can be expanded in terms of the
computational basis
\[
|\psi\rangle=\sum_{i_{1},i_{2},\cdots,i_{n}=0}^{d-1}c_{i_{1},i_{2}%
,\cdots,i_{n}}|i_{1},i_{2},\cdots,i_{n}\rangle=:\sum_{\eta\in\mathbb{N}%
_{d}^{\otimes n}}c_{\eta}|\eta\rangle\,,
\]
where a basis vector is denoted by $\eta=\left(  i_{1},i_{2},\cdots,i_{n}\right)  \in\mathbb{N}%
_{d}^{\otimes n}$.% We denote the computational number of $\eta$ as $\#\left(
%\eta\right)  :=\sum_{k}i_{k}d^{\left(  n-i_{k}\right)  }$.
 This vector notation will facilitate the upcoming derivations. A crucial element of the
notation in this paper will be the permutation operator acting upon two vectors,
exchanging vector components corresponding to the set of indices. E.g. the
permutation operator $P_{\{1,3\}}(\eta_1,\eta_2)$ will exchange the
first and third component of the vector $\eta_{1}$
with the corresponding component of the vector
$\eta_{2}$, i.e.
\[
P_{\left\{  \textcolor{red}{1},\textcolor{blue}{3}\right\}  }%
(\mathbf{\textcolor{red}{0}}1\mathbf{\textcolor{blue}{2}}%
13,\mathbf{\textcolor{red}{3}}0\mathbf{\textcolor{blue}{1}}%
21)=(\mathbf{\textcolor{red}{3}}1\mathbf{\textcolor{blue}{1}}%
13,\mathbf{\textcolor{red}{0}}0\mathbf{\textcolor{blue}{2}}21).
\]
Using this notation one can write down a very simple expression for the linear
entropy of a reduced state $\rho_{\gamma}$ (derivation see
section \ref{linear entropy of pure state} in the appendix)
\begin{equation}
S_{L}\left(  \rho_{\gamma}\right)  =\sum_{\eta_{1}\not =\eta_{2}%
}\left\vert c_{\eta_{1}}c_{\eta_{2}}-c_{\eta_{1}^{\gamma}}c_{\eta_{2}^{\gamma
}}\right\vert ^{2}, \label{result. linear entropy}%
\end{equation}
where $\left(  \eta_{1}^{\gamma},\eta_{2}^{\gamma}\right)  =P_{\gamma}\left(
\eta_{1},\eta_{2}\right)  $. %Let a state $\rho$ be decomposable into pure states
%$\left\vert \psi_{i}\right\rangle =\sum_{\eta}c_{\eta}^{i}\left\vert
%\eta\right\rangle $, then we have the explicit form of the GME measure%
%\begin{equation}
%E_{m}(\rho)=\inf_{\{p_{i},|\psi_{i}\rangle\}}\sum_{i}p_{i}\min_{\gamma}%
%\sqrt{\sum_{\eta_{1}\not =\eta_{2},\in\mathbb{N}_{d}^{\otimes n}}\left\vert
%c_{\eta_{1}}^{i}c_{\eta_{2}}^{i}-c_{\eta_{1}^{\gamma}}^{i}c_{\eta_{2}^{\gamma
%}}^{i}\right\vert ^{2}}\,. \label{result. explicit GME measure}%
%\end{equation}
For pure states we can of course find lower bounds on $E_m(|\psi\rangle\langle\psi|)$ by lower bounding the linear entropy for all possible bipartitions. For mixed states we can then provide a lower bound for the convex roof $E_m(\rho)$. We now illustrate our method in one exemplary case and then continue to articulate the main theorem.\\
Note that the linear entropy of subsystems has been widely used for lower bounding measures of entanglement due to the well known and simple structure of eq.(\ref{result. linear entropy}). None of the previous methods, however, work for lower bounding the inherently multipartite measure $E_m(\rho)$, due to the additional minimization over all bipartitions in each decomposition element of the convex roof.

\subsection{W-states\label{sec. W-state}}
In order to demonstrate how our framework works let us start by deriving the explicit lower bound detecting the three-qubit $W$ state $|W\rangle=\frac{1}{\sqrt{3}}(|001\rangle+|010\rangle+|100\rangle)$. For three-qubit states there are three bipartitions ($1|23,2|13,3|12$) and thus we have three linear entropies to look at in order to calculate $E_m(|\psi\rangle\langle\psi|)$,
\begin{gather}
\sqrt{S_L(\rho_1)}=\\2\sqrt{|c_{001}c_{100}-c_{101}c_{000}|^2+|c_{010}c_{100}-c_{110}c_{000}|^2+(\cdots)}\,,\nonumber\\
\sqrt{S_L(\rho_2)}=\\2\sqrt{|c_{010}c_{100}-c_{110}c_{000}|^2+|c_{010}c_{001}-c_{011}c_{000}|^2+(\cdots)}\,,\nonumber\\
\sqrt{S_L(\rho_3)}=\\2\sqrt{|c_{001}c_{100}-c_{101}c_{000}|^2+|c_{010}c_{001}-c_{011}c_{000}|^2+(\cdots)}\nonumber\,.
\end{gather}
Now using $\sqrt{a^2+b^2}\geq\frac{1}{\sqrt{2}}(a+b)$ (which is a specific case of the inequality \ref{eq. ineq. sqrt} in appendix \ref{appx. ineq.sqrt}) and $|a-b|\geq|a|-|b|$ it is obvious that
\begin{align}
\sqrt{S_L(\rho_1)}\geq\frac{2(|c_{001}c_{100}|-|c_{101}c_{000}|+|c_{010}c_{100}|-|c_{110}c_{000}|)}{\sqrt{2}}\,,\\
\sqrt{S_L(\rho_2)}\geq\frac{2(|c_{010}c_{100}|-|c_{110}c_{000}|+|c_{010}c_{001}|-|c_{011}c_{000}|)}{\sqrt{2}}\,,\\
\sqrt{S_l(\rho_3)}\geq\frac{2(|c_{001}c_{100}|-|c_{101}c_{000}|+|c_{010}c_{001}|-|c_{011}c_{000}|)}{\sqrt{2}}\,.
\end{align}
Then using $|ab|-\frac{1}{2}(a^2+b^2)\leq 0$ we can add one negative term for each entropy and it will still be a lower bound, i.e. we add $|c_{010}c_{001}|-\frac{1}{2}(|c_{010}|^2+|c_{001}|^2)$ in the first lower bound, $|c_{100}c_{001}|-\frac{1}{2}(|c_{100}|^2+|c_{001}|^2)$ in the second and $|c_{010}c_{100}|-\frac{1}{2}(|c_{010}|^2+|c_{100})|^2$ in the third. Then we can use that $\min[P-N_1,P-N_2,P-N_3]\geq P-N_1-N_2-N_3$ and end up with
\begin{align}
E_m(|\psi\rangle\langle\psi|)\geq\sqrt{2}(|c_{001}c_{100}|+|c_{001}c_{010}|+|c_{100}c_{010}|)-\nonumber\\
\frac{\sqrt{2}}{2}(|c_{010}|^2+|c_{100}|^2+|c_{001}|^2)-\nonumber\\
\sqrt{2}(|c_{101}c_{000}|+|c_{110}c_{000}|+|c_{011}c_{000}|)\,.
\end{align}
Finally we can bound the convex roof using the following two relations
\begin{align}
\inf_{\{p_i,|\psi_i\rangle\}}\sum_ip_i|c^i_{\eta_1} c^i_{\eta_2}|&\geq|\langle\eta_1|\rho|\eta_2\rangle|\,,\\
\inf_{\{p_i,|\psi_i\rangle\}}\sum_ip_i|c^i_{\eta_1} c^i_{\eta_2}|&\leq\sqrt{\langle\eta_1|\rho|\eta_1\rangle\langle\eta_2|\rho|\eta_2\rangle}\,,
\end{align}
and end up with a lower bound for mixed states as
\newpage
\begin{gather}
E_m(\rho)\geq\sqrt{2}(|\langle 001|\rho|100\rangle|+|\langle 001|\rho|010\rangle|+|\langle 100|\rho|010\rangle|)-\nonumber\\
\frac{\sqrt{2}}{2}(\langle 010|\rho|010\rangle+\langle 100|\rho|100\rangle+\langle 001|\rho|001\rangle)-\nonumber\\
\sqrt{2}\sqrt{\langle 101|\rho|101\rangle\langle 000|\rho|000\rangle}-\nonumber\\
\sqrt{2}\sqrt{\langle 110|\rho|110\rangle\langle 000|\rho|000\rangle}-\nonumber\\
\sqrt{2}\sqrt{\langle 011|\rho|011\rangle\langle 000|\rho|000\rangle}\,.
\label{WWitness}
\end{gather}
Surprisingly this leads directly to the nonlinear entanglement witness inequality presented in Refs.~\cite{Guehnewit,hmgh1} up to a factor of $\sqrt{2}$. Using only simple algebraic relations we have thus shown how to lower bound the convex roof construction. The first apparent strength of this lower bound is the limited number of density matrix elements needed to compute it. E.g. in our exemplary three-qubit case only ten out of possibly sixty-four elements need to be measured. Obviously we can extend the analysis using the same techniques to systems beyond three qubits.

\section{A General Construction of lower bounds on the GME measure $E_{m}$}
Now we can generalize the connection of the 3-qubit W state witness and the measure $E_{m}$.
%First we provide a general theorem how such lower bounds (and thus of course witnesses) can be obtained
%for any given pure state.
Just as for three qubits we can always get lower bounds by summing the
coefficient pairs $c_{\eta_1}c_{\eta_2}$ that belong to a certain target pure state and appear in some or all reduced linear entropies. The construction of such general lower bounds also starts by selecting a subset of coefficient pairs that will be translated into off-diagonal elements $\rho_{\eta_1,\eta_2}$, where $(\eta_1,\eta_2)$ is the vector basis pair denoting the row and column of the element in density matrix $\rho$. We denote the selected vector basis pairs as $R:=\{(\eta_1,\eta_2)\}$. Then we can repeat the steps analogously to eq.(6-11) and arrive at a general lower bound on the measure as the following theorem:
%We can even write down this construction in a very general form in the following theorem:
\newpage
\begin{widetext}

\begin{theorem}
[A general lower bound on the GME measure]%
\label{theorem. lower bound of GME measure}For a set of row-column pairs $R=\{(\eta_1,\eta_2)\}$,
the genuine multipartite entanglement measure $E_{m}$ has the following
lower bound:
\begin{equation}
E_{m}\geq 2\sqrt{\frac{1}{\left\vert R\right\vert -N_{R}}}\left[  \sum_{\left(
\eta_{1},\eta_{2}\right)  \in R}\left(  \left\vert \rho_{\eta_{1}\eta_{2}%
}\right\vert -\sum_{\gamma\in\Gamma(\eta_{1},\eta_{2})}\sqrt{\rho_{\eta
_{1}^{\gamma}\eta_{1}^{\gamma}}\rho_{\eta_{2}^{\gamma}\eta_{2}^{\gamma}}%
}\right)  -\left(  \frac{1}{2}\sum_{\eta\in I(R)}N_{\eta}\left\vert \rho
_{\eta\eta}\right\vert \right)  \right] \label{eq. theorem-result}\,.
\end{equation}
The right-hand-side of eq.(\ref{eq. theorem-result}) defines a GME witness $W_R(\rho)$,
where $\rho_{\eta_{1},\eta_{2}}:=\langle\eta_{1}|\rho|\eta_{2}\rangle$,
$(\eta_{1}^{\gamma},\eta_{2}^{\gamma}):=P_{\gamma}(\eta_{1},\eta_{2})$,
$\Gamma(\eta_{1},\eta_{2}):=\left\{  \gamma:(\eta_{1}^{\gamma},\eta
_{2}^{\gamma})\notin R\right\}  $ and $I(R):=\left\{  \eta:\exists\eta
^{\prime}\text{ that }(\eta^{\prime},\eta)\text{ or }(\eta,\eta^{\prime})\in
R\right\}  $ is the set of basis vectors $\eta$, which appear in the
set $R$.

$N_{R}$ is the maximal (or minimal) value of $|R^{\gamma}|$ over all possible bipartitions $\gamma|\bar{\gamma}$,
where $R^{\gamma}$ is the set of coefficient pairs $(c_{\eta_{1}},c_{\eta_{2}})\in R$, which do
not contribute to the $\gamma$-subsystem entropy.

$N_{\eta}$ are normalization constants given
by the maximal value of $n_{\eta}^{\gamma}$ over all possible bipartitions $\gamma|\bar{\gamma}$, where $n_{\eta}^{\gamma}$ is the number of coefficients $c_\eta$ from some pairs in $R$, which are not counted in the $\gamma$-subsystem entropy (and how many are counted depends on whether one chooses $N_R$ to be maximal or minimal). \\
\end{theorem}
\begin{proof}
See Appendix \ref{appx. proof of lower bound on GME measure} for the full proof.
\end{proof}
\end{widetext}
It is evident that not every choice of coefficient pairs will yield a useful lower bound, because one really needs to select those that are actually contributing to multipartite entanglement. There is however always an obvious choice. The set of coefficient pairs $R$ must be chosen such that in every subsystem at least one of the elements of $R$ contribute to the linear entropy of the reduced state. E.g. in the case of GHZ states given in a specific basis $|GHZ\rangle=\frac{1}{\sqrt{2}}(|0\rangle^{\otimes n}+|1\rangle^{\otimes n})$ one would choose the pair $({00\cdots0},{11\cdots1})$, which contributes to all reduced entropies. In the general case however there is still some freedom of choice left to get a valid lower bound. For some sets $R$ it can happen, that the coefficients do not contribute to every subsystem entropy equally (which we show in an exemplary case in section \ref{sec. four-qubit-example}). Then one can choose $N_R$ in different ways, but in all considered cases we found that choosing it maximal or minimal will produce the best bounds (where choosing it maximal usually yields the tightest bounds close to pure states, whereas choosing it minimal improves the noise resistance). Since these coefficients are in general basis dependent, so is also our witness construction. The prefactor $\sqrt{\frac{1}{\left\vert R\right\vert-N_R }}$ suggests that the optimal basis for constructing such a lower bound is given by the minimal tensor rank representation of the pure state.
\section{Applications and Examples}
\subsection{Four-qubit singlet state}\label{sec. four-qubit-example}
Let us illustrate how to apply Theorem \ref{theorem. lower bound of GME measure} with an explicit example.
In an experimental setting where one expects to produce a four-qubit singlet state
(which was e.g. discussed in the context of solving the liar detection problem in Ref.~\cite{cabello}), i.e.
\begin{align}
|S_4\rangle=&\frac{1}{2\sqrt{3}}(2|0011\rangle+2|1100\rangle-|0110\rangle\nonumber\\
            &-|1001\rangle-|1010\rangle-|0101\rangle)\,,
\end{align}
one is confronted with the following expected coefficients: $c_{0011},c_{1100},c_{0101},c_{1010},c_{0110},c_{1001}$. Following the recipe of theorem \ref{theorem. lower bound of GME measure} we now select some coefficient pairs. We could choose e.g. $R_1=(0011,0101)$, $R_2=(0011,1010)$, $R_3=(0011,0110)$ and $R_4=(0011,1001)$, such that $R=\{R_1,R_2,R_3,R_4\}$. For this selection we use theorem \ref{theorem. lower bound of GME measure} to bound the GME measure. We see that in every subsystem at least two of these pairs appear naturally. Although there are more coefficient pairs we now choose to only take into account two per subsystem entropy and thus choose $N_R$ to be the minimal number of coefficient pairs in every subsystem which gives $N_R=2$. Thus we need to add negative terms that compensate for the missing terms just as we did in the three-qubit case, but now we need to do it two times in every subsystem. This results in the following individual prefactors $N_\eta$ for the diagonal elements: $N_{0011}=2$ (as this coefficient appears in two missing pairs in every subsystem), $N_{0101}=1$, $N_{1001}=1$, $N_{1010}=1$ and $N_{0110}=1$ (as those appear maximally once per subsystem entropy). Inserting this in theorem \ref{theorem. lower bound of GME measure} we end up with the lower bound as
\begin{gather}
E_m(\rho)\geq\frac{2}{\sqrt{2}}(|\rho_{R_1}|+|\rho_{R_2}|+|\rho_{R_3}|+|\rho_{R_4}|\nonumber\\
-\sqrt{\rho_{0111,0111}\rho_{0001,0001}}-\sqrt{\rho_{0111,0111}\rho_{0010,0010}}\nonumber\\
-\sqrt{\rho_{1011,1011}\rho_{0001,0001}}-\sqrt{\rho_{1011,1011}\rho_{0010,0010}}\nonumber\\
-\frac{1}{2}(\rho_{0101,0101}+\rho_{1001,1001}+\rho_{1010,1010}+\rho_{0110,0110})\nonumber\\
-\rho_{0011,0011})\,.
\end{gather}
We have thus created a nonlinear witness function that lower bounds our measure. From an experimental point of view this is very favorable as few local measurement settings suffice to ascertain the needed thirteen density matrix elements (especially since the nine diagonal elements can be constructed from a single measurement setting). Of course we could also exploit the connection of our lower bound to the Dicke state witness $Q^{(2)}_2$ (which is discussed in section \ref{Dicke}), which also detects GME in this state (although at the cost of more required measurements). In this case even the resistance to white noise is more favorable with our construction method, as for a state $\rho=p|S_4\rangle\langle S_4|+\frac{1-p}{16}\mathbbm{1}$ this exemplary lower bound detects GME until $p=\frac{21}{29}\approx 0.72$, whereas the old witness construction yields a worse resistance up to $p=\frac{27}{35}\approx 0.77$. This shows the versatility of our general approach. By choosing certain coefficients one can tailor these lower bounds to specific experimental situations. If one is confronted with a low noise system it is always beneficial to choose as few coefficients as possible, such that very few local measurements suffice (even a number that is linear in the size of the system is often sufficient). Every additional measurement can then be included in the lower bound and improves the bound and its noise resistance if necessary.

\subsection{Bipartite witnesses and lower bounds on the measure}
Although we have presented our theorem and measures in the general case of $n$-qudits, we can always apply the lower bounds also for $n=2$, as our theorem holds for any $n$ and $d$.	Suppose we are given a bipartite qutrit system and want to lower bound the concurrence with only a few local measurements.
	If the expected state is e.g. $|\psi\rangle=\frac{1}{\sqrt{3}}(|00\rangle+|11\rangle+|22\rangle)$ we can use the lower bounding procedure outlined above, yielding
\begin{align}
	E_m(\rho)\geq\frac{2}{\sqrt{3}}(\Re e[\left.\left\langle 00| \rho |11 \right\rangle\right.]-\sqrt{\left.\left\langle 01| \rho| 01 \right\rangle\left\langle| 10 \rho |10 \right\rangle\right.}+\nonumber\\
	\Re e[\left.\left\langle 00| \rho |22 \right\rangle\right.]-\sqrt{\left.\left\langle 02| \rho |02 \right\rangle\left\langle 20| \rho| 20 \right\rangle\right.}+\nonumber\\
	\Re e[\left.\left\langle 11| \rho |22 \right\rangle\right.]-\sqrt{\left.\left\langle 12| \rho |12\right\rangle\left\langle 21| \rho |21 \right\rangle\right.	})\,.
	\end{align}
	
	In order to determine the lower bound we have to measure nine different density matrix elements. Of course any density matrix element can always be obtained via local measurements. How these measurements can be performed in a basis consisting of a tensor product of the generalized Gell-Mann matrices we show explicitly in appendix \ref{GellMann}.\\
It turns out that these nine different density matrix elements can be obtained via ten local measurement settings. Let us study the lower bound in the presence of noise. Suppose we have white noise in the system, i.e. $\rho=p|\psi\rangle\langle\psi|+\frac{1-p}{d}\id$. Calculating the lower bound results in $E_m(\rho)\geq\frac{2(4p-1)}{\sqrt{27}}$, which is equivalent to the analytical expression of Wootter's concurrence for these systems (as proven in Ref.~\cite{hashemi, caves}). In this case we have a necessary and sufficient entanglement criterion and a tight lower bound on the concurrence from ten local measurements for a special class of states. Indeed if one generalizes this example to arbitrary dimension $d$, we find that the bound is always tight for bipartite isotropic states.

\subsection{Dicke States}\label{Dicke}
We will now continue to show how this construction relates to an entanglement witness for Dicke-state, which are multi-dimensional generalizations of the $W$ states(which were first introduced in the context of laser emission in Ref.~\cite{Dicke}).

In the original article \cite{maetal}, where this approach was first
introduced, the authors connected the violation of a witness suitable for GHZ
states (first introduced in Ref.~\cite{Guehnewit} and later presented in a
more general framework in Ref.~\cite{hmgh1}) with a lower bound on the measure
$E_{m}$. We want to follow this approach and establish a general connection
between a set of witnesses suitable for all generalized Dicke states
introduced in Ref.~\cite{hesgh1} and generalized in Ref.~\cite{shgh1}. To that
end let us first introduce a concise notation for those states.\newline Let
$\alpha$ be a set containing specific subsystems of a multipartite state. We
then define the state $|\alpha^{l}\rangle$ as a tensor product of states
$|l\rangle$ for all subsystems not contained in $\alpha$ and excited states $|l+1\rangle$
in the subsystems contained in $\alpha$. E.g. for the four-partite state
$|\{1,3\}^{2}\rangle$ we have $|3232\rangle$. Using this abbreviated notation
we can define a generalized set of  Dicke states, consisting of $n$ $d$-dimensional subsystems,
as
\begin{align}
|D_{m}^{d}\rangle=\frac{1}{\sqrt{{n\choose m}(d-1)}}\sum_{l=0}^{d-2}\sum_{\alpha
:|\alpha|=m}|\alpha^{l}\rangle\,,
\end{align}
where the parameter $m$ denotes the number of excitations, with $0<m<n$.

Since the explicit form of the nonlinear witness from Ref.~\cite{shgh1} will be used in the following
considerations we will repeat it in appendix \ref{DickeWitness}.
For all biseparable states this witness $Q_{m}^{\left(  d\right)  }$ is
strictly smaller equal zero, i.e.
\begin{align*}
Q\left(  \rho\right)   &  \leq0\Leftarrow\rho\text{ is biseparable}\\
Q\left(  \rho\right)   &  >0\Rightarrow\rho\text{ is multipartite
entangled}\,.
\end{align*}

Furthermore, the witness can also detect the ``dimensionality'' of GME, by which we mean
the maximal number of degrees of freedom $f_\rho (f_\rho \le d)$
that occurs in the  pure states of an ensemble constituting $\rho$, minimised over all ensembles
(this is the natural generalization of the concept of Schmidt number
\cite{schmidtnumber}
to multipartite systems, further explored e.g. in Ref.\cite{shgh1}). I.e. the dimensionality is defined as
\begin{align}
f_\rho:=\inf_{\{p_i,|\psi_i\rangle\}}\max_i(\min_\gamma(\text{rank}(\rho_\gamma)))
\end{align}
Since
\begin{align*}
Q_{m}^{\left(  d\right)  }\left(  \rho\right)   &  \leq f_{\rho}-1,\,\forall
\rho\,,
\end{align*}%
we can directly infer that
\[
Q_{m}^{\left(  d\right)  }\left(  \rho\right)  >f-2\Rightarrow f_\rho\geq f%
\]
In  fig.\ref{fig. f_dimensional_gme_entanglement} we show how
$Q_{m}^{\left(  d\right)  }$ detects the GME dimensionality.
The maximal violation of these inequalities is always achieved for
$m$-excitation Dicke states, i.e. $Q_{m}^{\left(  d\right)  }(|D_{m}^{d}
\rangle\langle D_{m}^{d}|)=d-1$.%
\begin{figure}
[ptb]
\begin{center}
\includegraphics[
natheight=7.499600in,
natwidth=9.999800in,
height=2.6541in,
width=3.5336in
]%
%{f_dimensional_GME_entanglement.eps}%
{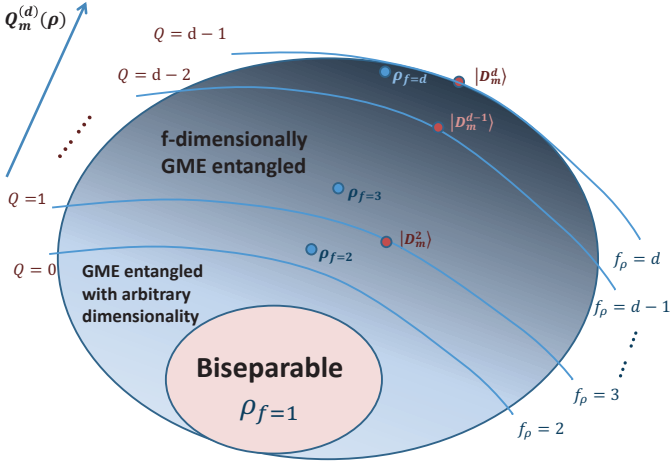}%
\caption{(Color online) The witness $Q_{m}^{\left(  d\right)  }$ can quantify the genuine
multipartite entanglement with the GME dimensionality $f_\rho$, where $f_\rho \le d$ . A state is
biseparable iff $f_\rho=1$ and GM entangled iff $f_\rho\geq2$. For Dicke states it holds
$Q_{m}^{\left(  d\right)  }\left(  \left\vert D_{m}^{f}\right\rangle
\left\langle D_{m}^{f}\right\vert \right)  =f-1$.}%
\label{fig. f_dimensional_gme_entanglement}%
\end{center}
\end{figure}
%EndExpansion

\bigskip

If we can find a proper $R$, as a result of
theorem \ref{theorem. lower bound of GME measure} that uses the Dicke state
coefficients, we can connect a lower bound of the measure $E_{m}$ with the GME
witness $Q_{m}^{\left(  d\right)  }\,\left(  \rho\right)  $. Indeed choosing
the ordered subset $R_{\sigma}$ of the set of coefficients $\sigma$ used in
(\ref{eq.: GME-witness formular}), i.e.%
\[
R_{\sigma}=\left\{  \left(  \alpha^{a},\beta^{b}\right)  \in\sigma:a\leq
b\right\}\,,
\]
we immediately arrive at a lower bound on $E_{m}$ as
\begin{equation}
E_{m}\left(  \rho\right)
\geq m\sqrt{\frac{1}{\left\vert R_{\sigma} \right\vert - N_{R_{\sigma}} }%
}Q_{m}^{\left(  d\right)  }\,\left(  \rho\right)
\geq m\sqrt{\frac{1}{\left\vert R_{\sigma} \right\vert }%
}Q_{m}^{\left(  d\right)  }\,\left(  \rho\right)
,\label{result. lower bound on GME measure}%
\end{equation}
where $\left\vert R_{\sigma}\right\vert =\frac{1}{2}\left(  d-1\right)
^{2}{\binom{n}{m}m}(n-m)$. In this case $N_{\eta}\leq m\left(  n-m-1\right)
+\Theta\left(  d-3\right)  \left(  n-m\right)  $, where $\Theta$ is a Heaviside step function.

\subsection{PPT-Witness and Our Witness}
\begin{figure}[ht!]
\begin{center}
\includegraphics[
natheight=7.158100in,
natwidth=8.496800in,
height=2.4967in,
width=2.9611in
]%
%{ppt_witness_and_vs_witness.eps}%
{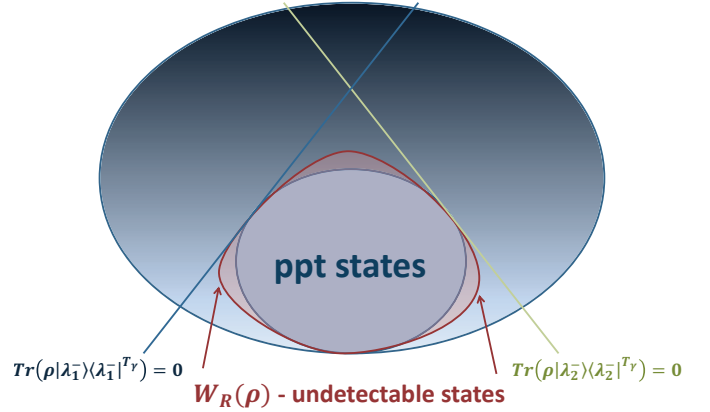}%
\caption{(Color online) PPT witness compared to our lower bounds (given in terms of the nonlinear witness $W_{R}\left(  \rho\right)  $): The set of $W_{R}\left(  \rho\right)  $
undetectable states denotes the set that is not detected by one specific $W_R(\rho)$ and is strictly larger than the set of PPT-states. However the set of states detected by $W_{R}\left(  \rho\right)  $ is strictly larger than the set detected by any standard PPT-witness.
}%
\label{fig. ppt_witness_and_vs_witness}%
\end{center}
\end{figure}

Using the result on entanglement across bipartitions from the previous section we can explore the relation of our lower bounds to other bipartite entanglement witnesses.
In our witness construction, the permutation operator $P_{\gamma}$ acting on a pure state is a $\gamma|\bar{\gamma}$-partial transpose operator, i.e. $P_\gamma|\psi\rangle\langle\psi|=(|\psi\rangle\langle\psi|)^{T_{\gamma|\bar{\gamma}}}$ (in the sense that our permutation operator now acts upon the index pairs of the coefficients of the pure state). It is thus intuitive to believe that there is certain connection between our witness and  a PPT-witness \cite{pptwitness}. Indeed our witnesses are related to a standard PPT-witness construction (where the witnesses separate the convex set of states that are positive under partial transpose (PPT) from its complement). E.g. for diagonal GHZ states we can use the standard PPT-witness construction which goes as follows. For
$\left\vert \text{GHZ}_{\eta_{1},\eta_{2}}\right\rangle :=\frac{1}{\sqrt{2}}\left(
\left\vert \eta_{1}\right\rangle +\left\vert \eta_{2}\right\rangle \right)  $
with $\eta_{1}+\eta_{2}=\left(  d-1,\cdots,d-1\right)  $, we can use the eigenvector belonging to the negative eigenvalue
of the $\gamma|\bar{\gamma}$-partial transposed $\left\vert \text{GHZ}_{\eta_{1}%
,\eta_{2}}\right\rangle \left\langle \text{GHZ}_{\eta_{1},\eta_{2}}\right\vert
^{T_{\gamma}}$ which we denote as $\left\vert \lambda_{\eta_{1},\eta_{2}}^{-}\right\rangle
=\frac{1}{\sqrt{2}}\left(  \left\vert \eta_{1}^{\gamma}\right\rangle
-\left\vert \eta_{2}^{\gamma}\right\rangle \right)  $. One can then construct
the PPT-witness and write its expectation value as
\begin{equation}
\Omega_{\text{ppt}}^{\gamma|\bar{\gamma}}\left(  \rho,\left\vert \lambda_{\eta
_{1},\eta_{2}}^{-}\right\rangle \right)  =\text{Tr}\left(  \left\vert \lambda
_{\eta_{1},\eta_{2}}^{-}\right\rangle \left\langle \lambda_{\eta_{1},\eta_{2}%
}^{-}\right\vert ^{T_{\gamma|\bar{\gamma}}}\rho\right)
,\label{eq. ppt-witness construction}%
\end{equation}
For instance in the three-qubit case, 
\begin{align}
\left\vert \lambda_{001,110}^{-}\right\rangle
\left\langle \lambda_{001,110}^{-}\right\vert ^{T_{1|23}}=\nonumber
\end{align}
\begin{align}
\frac{1}{2}\left(
\begin{array}
[c]{cccccccc}%
0 & 0 &  & \cdots & \cdots &  & 0 & 0\\
& 0 &  &  &  &  & -1 & \\
&  & 1 &  &  & 0 &  & \\
\vdots &  &  & 0 & 0 &  &  & \vdots\\
\vdots &  &  & 0 & 0 &  &  & \vdots\\
&  & 0 &  &  & 1 &  & \\
& -1 &  &  &  &  & 0 & \\
0 &  &  & \cdots & \cdots &  &  & 0
\end{array}
\right)  .
\end{align}
With the PPT-witness construction in eq.(\ref{eq. ppt-witness construction})
we end up with the following PPT-witness expectation value%
\begin{align}
\Omega_{\text{ppt}}^{\gamma|\bar{\gamma}}\left(  \rho,\left\vert \lambda_{GHZ}%
^{-}\right\rangle \right)  =\frac{1}{2}\left(  \rho_{\eta_{1}^{\gamma}\eta
_{1}^{\gamma}}+\rho_{\eta_{2}^{\gamma}\eta_{2}^{\gamma}}\right)
-\operatorname{Re}\left(  \rho_{\eta_{1}\eta_{2}}\right)  .
\end{align}
Under the fixed  bipartition $\gamma|\bar{\gamma}$, we construct our witness by
choosing $R=\left(  \eta_{1},\eta_{2}\right)  $ as%
\begin{equation}
-W_{\left(  \eta_{1},\eta_{2}\right)  }^{\gamma}\left(  \rho\right)
=\sqrt{\rho_{\eta_{1}^{\gamma}\eta_{1}^{\gamma}}\rho_{\eta_{2}^{\gamma}%
\eta_{2}^{\gamma}}}-\left\vert \rho_{\eta_{1}\eta_{2}}\right\vert\,.
\end{equation}
It is obvious that $-W_{\left(  \eta_{1},\eta_{2}\right)  }^{\gamma}\left(
\rho\right)  $ $\leq\Omega_{\text{ppt}}^{\gamma|\bar{\gamma}}\left(  \rho,\left\vert
\lambda_{GHZ}\right\rangle \right)  $. Hence we say that the witness
$W_{R}\left(  \rho\right)  $ is stronger than the PPT-witness $\Omega
_{\text{ppt}}^{\gamma|\bar{\gamma}}\left(  \rho,\left\vert \lambda_{\eta_{1},\eta
_{2}}^{-}\right\rangle \right)  $.

 The relation between our witness, the PPT-witness
and the PPT-convex set is illustrated in
fig.\ref{fig. ppt_witness_and_vs_witness}. For clearness we just draw two
PPT-witnesses in the figure.
For the $n$-qudit case there are $\frac{1}{2}d^{n}$ such eigenvectors $\left\vert
\lambda_{\eta_{1},\eta_{2}}^{-}\right\rangle $, corresponding to negative eigenvalues. Every witness $\Omega
_{\text{ppt}}^{\gamma|\bar{\gamma}}\left(  \rho,\left\vert \lambda_{\eta_{1},\eta
_{2}}^{-}\right\rangle \right)  $ is tangent to the set of PPT states (i.e. there exists one PPT state for which the witness yields zero). However also our witness $W_{R}\left(  \rho\right)$ is zero for all these PPT states, i.e.
our new witness detects more states than the traditional PPT-witness.
\section{Conclusions}
In conclusion we have presented a method to derive lower bounds on a measure of genuine multipartite entanglement. We show that in experimentally plausible scenarios (i.e. one knows which state one aims to produce) we can derive such lower bounds simply based on coefficients of the corresponding pure states. We also connected the lower bound construction to a framework of nonlinear entanglement witnesses developed in Refs.~\cite{Guehnewit,hmgh1,Dicke,hesgh1,shgh1}. These witnesses are experimentally feasible in terms of required local measurement settings. We provide further evidence in the bipartite case, where we also show that for certain families of mixed states our lower bounds are tight.\\
Some open questions remain, such as whether this general construction method will work for all kinds of states and how it can be generalized beyond just multi- and bipartite entanglement, but anything in between. We want to point out that recently also other authors have used a similar approach to bound this measure in the bipartite case \cite{hashemi} and for multipartite $W$ states \cite{severinima}.

{\em Acknowledgements:} We thank the QCI group in Bristol and Matteo Rossi for discussions.
JYW, HK and DB acknowledge financial support of DFG (Deutsche Forschungsgemeinschaft).
MH gratefully acknowledges support from the EC-project IP "Q-Essence" and the ERC Advanced
Grant "IRQUAT".

%%%%%%%%%%%%%%%%%%%%%%%%%%%%%%%% Appendix %%%%%%%%%%%%%%%%%%%%%%%%%%%%%%%%
\newpage\appendix*

\bigskip\clearpage

%\footnotesize

\section{Proofs\label{appx. proof of result}}

\subsection{The Formulas Used in The Main Article}

\subsubsection{Reduced linear entropy of pure states}\label{linear entropy of pure state}
Let $\left\vert \psi\right\rangle =\sum_{\eta\in\mathbb{N}_{d}^{\otimes n}}%
c_{\eta}\left\vert \eta\right\rangle $ be an n-qudit pure state. The linear
entropy of $\left\vert \psi\right\rangle $ can be written as%
\begin{equation}
S_{L}\left(  \rho_{\gamma}\right)  =\sum_{\eta_{1}\not =\eta
_{2},\in\mathbb{N}_{d^{n}}} %
\left\vert c_{\eta_{1}}c_{\eta_{2}}-c_{\eta_{1}^{\gamma}}c_{\eta_{2}^{\gamma
}}\right\vert ^{2},
\end{equation}
where $\left(  \eta_{1}^{\gamma},\eta_{2}^{\gamma}\right)  =P_{\gamma}\left(
\eta_{1},\eta_{2}\right)  $.

\begin{proof}
The linear entropy regarding a specific partition $\gamma|\bar{\gamma}$ is
defined as $S_{L}\left(  \rho_{\gamma}\right)  =2(1-tr\left(  \rho_{\gamma}%
^{2}\right))$, where $\rho_{\gamma}$ is the $\gamma$-reduced matrix of $\rho
$. The trace of $\rho_{\gamma}$ is $tr\left(  \rho_{\gamma}^{2}\right)
=\sum_{\alpha_{1},\alpha_{2}\in H_{\gamma}}\left(  \rho_{\gamma}\right)
_{\alpha_{1}\alpha_{2}}\left(  \rho_{\gamma}\right)  _{\alpha_{2}\alpha_{1}}$,
where $H_{\gamma}$ is the subspace of the reduction $\gamma$. We separate the summation
into diagonal and off-diagonal parts. For the diagonal part we use the
normalization condition to evaluate its value.%
\begin{align}
&  tr\left(  \rho_{\gamma}^{2}\right)  \nonumber\\
&  =\sum_{\alpha_{1}=\alpha_{2}}\left(  \rho_{\gamma}\right)  _{\alpha
_{1}\alpha_{1}}^{2}+\sum_{\alpha_{1}\not =\alpha_{2}}\left(  \rho_{\gamma
}\right)  _{\alpha_{1}\alpha_{2}}\left(  \rho_{\gamma}\right)  _{\alpha
_{2}\alpha_{1}}\nonumber\\
&  =\left(  \sum_{\alpha}(\rho_{\gamma})_{\alpha\alpha}\right)  ^{2}%
-\sum_{\alpha_{1}\not =\alpha_{2}}\left(  \rho_{\gamma}\right)  _{\alpha
_{1}\alpha_{1}}\left(  \rho_{\gamma}\right)  _{\alpha_{2}\alpha_{2}%
}\nonumber\\
&  +\sum_{\alpha_{1}\not =\alpha_{2}}\left(  \rho_{\gamma}\right)
_{\alpha_{1}\alpha_{2}}\left(  \rho_{\gamma}\right)  _{\alpha_{2}\alpha_{1}%
}\nonumber\\
&  =1-\sum_{\substack{\alpha_{1}\not =\alpha_{2}\in H_{\gamma}\\\beta
_{1},\beta_{2}\in H_{\bar{\gamma}}}}\left\vert c_{\alpha_{1}\otimes\beta_{1}%
}\right\vert ^{2}\left\vert c_{\alpha_{2}\otimes\beta_{2}}\right\vert
^{2}\nonumber\\
&  +\sum_{\substack{\alpha_{1}\not =\alpha_{2}\in H_{\gamma}\\\beta_{1}%
,\beta_{2}\in H_{\bar{\gamma}}}}c_{\alpha_{1}\otimes\beta_{1}}c_{\alpha
_{2}\otimes\beta_{1}}^{\ast}c_{\alpha_{2}\otimes\beta_{2}}c_{\alpha_{1}%
\otimes\beta_{2}}^{\ast}.
\end{align}
By exchanging the indices $\alpha_{1}$ and $\alpha_{2}$ one has
\begin{align}
tr\left(  \rho_{\gamma}^{2}\right)   &  =1-\frac{1}{2}\sum_{\substack{\alpha
_{1}\not =\alpha_{2}\in H_{\gamma}\\\beta_{1},\beta_{2}\in H_{\bar{\gamma}}%
}}\left\vert c_{\alpha_{1}\otimes\beta_{1}}c_{\alpha_{2}\otimes\beta_{2}%
}-c_{\alpha_{1}\otimes\beta_{2}}c_{\alpha_{2}\otimes\beta_{1}}\right\vert
^{2}\nonumber\\
&  =1-\frac{1}{2}\sum_{\eta_{1},\eta_{2}\in\mathbb{N}_{d^{n}}}\left\vert
c_{\eta_{1}}c_{\eta_{2}}-c_{\eta_{1}^{\gamma}}c_{\eta_{2}^{\gamma}}\right\vert
^{2},
\end{align}
where $\eta=\alpha\otimes\beta$ and $\left(  \eta_{1}^{\gamma},\eta
_{2}^{\gamma}\right)  =P_{\gamma}\left(  \eta_{1},\eta_{2}\right)  $. The
linear entropy is then calculated to%
\begin{equation}
S_{L}\left(  \rho_{\gamma}\right)  =\sum_{\eta_{1}\not =\eta_{2},\in
\mathbb{N}_{d^{n}}}\left\vert c_{\eta_{1}}c_{\eta_{2}}-c_{\eta_{1}^{\gamma}%
}c_{\eta_{2}^{\gamma}}\right\vert ^{2}.
\end{equation}

\end{proof}

\subsubsection{An Important Inequality}
\label{appx. ineq.sqrt}
The following is an inequality, which is crucial for derivation of the prefactor
$\sqrt{\frac{1}{|R|-N_R}}$ in the theorem \ref{theorem. lower bound of GME measure}:
\begin{equation}
\left\vert I\right\vert \sum_{i\in I}\left\vert a_{i}\right\vert ^{2}%
\geq\left\vert \sum_{i\in I}a_{i}\right\vert ^{2}. \label{eq. ineq. sqrt}%
\end{equation}

\begin{proof}
We prove this inequality by constructing two vectors as follows (using $|I|=n$)%
\begin{equation}
\vec{x}=\left(
\begin{array}
[c]{c}%
\left.
\begin{array}
[c]{c}%
a_{1}\\
\vdots\\
a_{1}%
\end{array}
\right\}  \text{n times}\\
\vdots\\
\left.
\begin{array}
[c]{c}%
a_{n}\\
\vdots\\
a_{n}%
\end{array}
\right\}  \text{n times}%
\end{array}
\right)  ,\vec{y}=\left(
\begin{array}
[c]{c}%
\begin{array}
[c]{c}%
a_{1}^{\ast}\\
\vdots\\
a_{n}^{\ast}%
\end{array}
\\
\vdots\\%
\begin{array}
[c]{c}%
a_{1}^{\ast}\\
\vdots\\
a_{n}^{\ast}%
\end{array}
\end{array}
\right)  .
\end{equation}
The right hand side of \ref{eq. ineq. sqrt} can be written as the scalar
product of $\vec{x}$ and $\vec{y}$.%
\begin{equation}
\left\vert \sum_{i\in I}a_{i}\right\vert ^{2}=\sum_{i,j\in I}a_{i}a_{j}^{\ast
}=\left\vert \vec{x}\cdot\vec{y}\right\vert .
\end{equation}
According to the Cauchy-Schwarz inequality, one can derive%
\begin{equation}
\left\vert I\right\vert \sum a_{i}^{2}=\left\vert \vec{x}\right\vert
\cdot\left\vert \vec{y}\right\vert \geq\left\vert \vec{x}\cdot\vec
{y}\right\vert =\left\vert \sum_{i\in I}a_{i}\right\vert ^{2} .%
\end{equation}

\end{proof}

\subsection{Proof of Theorem \ref{theorem. lower bound of GME measure} and
Approach of Construction of a GME Witness}
\label{appx. proof of lower bound on GME measure}

Firstly one can estimate the lower bound on $S_{L}\left(  \rho_{\gamma
}\right)  $ by summing its elements over a selected Region $R$, and dropping
the other non-negative summands (i.e. lower bounding them with $0$),%
\begin{align}
S_{L}\left(  \rho_{\gamma}^{i}\right)   &  \geq4\sum_{\left(  \eta_{1}%
,\eta_{2}\right)  \in R}\left\vert c_{\eta_{1}}^{i}c_{\eta_{2}}^{i}%
-c_{\eta_{1}^{\gamma}}^{i}c_{\eta_{2}^{\gamma}}^{i}\right\vert ^{2}%
\label{eq.: linear_entropy_lower_bound_1}\\
&  =4\sum_{\left(  \eta_{1},\eta_{2}\right)  \in R\backslash R^{\gamma}%
}\left\vert c_{\eta_{1}}^{i}c_{\eta_{2}}^{i}-c_{\eta_{1}^{\gamma}}^{i}%
c_{\eta_{2}^{\gamma}}^{i}\right\vert ^{2}.\nonumber
\end{align}
Here we add a prefactor $4$ in eq.(\ref{eq.: linear_entropy_lower_bound_1}),
since the symmetric factor of all $(\eta_{1},\eta_{2})$ equals $4$. That means
for every $(\eta_{1},\eta_{2})$ there are three other $\left(  \tilde{\eta
}_{1},\tilde{\eta}_{2}\right)  $ having the same value of $|c_{\tilde{\eta
}_{1}}c_{\tilde{\eta}_{2}}-c_{\tilde{\eta}_{1}^{\gamma}}c_{\tilde{\eta}%
_{2}^{\gamma}}|$ as $(\eta_{1},\eta_{2})$. Here we choose a non-degenerate
vector basis set $R$, and therefore need a prefactor $4$ in the lower bound.
The set $R^{\gamma}$ is the subset of $R$, whose elements do not contribute to
the linear entropy, i.e. $R^{\gamma}:=\left\{  \left(  \eta_{1},\eta
_{2}\right)  \in R:(\eta_{1}^{\gamma},\eta_{2}^{\gamma})=\left(  \eta_{1}%
,\eta_{2}\right)  \text{ or }\left(  \eta_{2},\eta_{1}\right)  \right\}  $.
Now we use the inequality (\ref{eq. ineq. sqrt}) to bound the square root of
$S_{L}\left(  \rho_{\gamma}^{i}\right)  $.%
\begin{align}
S_{L}\left(  \rho_{\gamma}^{i}\right)   &  \geq\frac{4}{\left\vert R\backslash
R^{\gamma}\right\vert }\left(  \sum_{\eta_{1},\eta_{2}\in R}\left\vert
c_{\eta_{1}}^{i}c_{\eta_{2}}^{i}-P_{\gamma}c_{\eta_{1}}^{i}c_{\eta}%
^{i}\right\vert \right)  ^{2}\label{eq. ineq._for_linear_entropy},\\
&  \Downarrow\nonumber\\
\sqrt{S_{L}\left(  \rho_{\gamma}^{i}\right)  } &  \geq2\sqrt{\frac
{1}{\left\vert R\right\vert -\left\vert R^{\gamma}\right\vert }}\sum_{\left(
\eta_{1},\eta_{2}\right)  \in R}\left\vert c_{\eta_{1}}^{i}c_{\eta_{2}}%
^{i}-c_{\eta_{1}^{\gamma}}^{i}c_{\eta_{2}^{\gamma}}^{i}\right\vert
.\label{result. lower bound on linear entropy}%
\end{align}
According to eq.(\ref{definitionroof}) together with
eq.(\ref{result. lower bound on linear entropy}), the lower bound reads {\footnotesize {
\begin{align}
&  E_{m}\geq\nonumber\\
&  2\inf_{\left\{  p_{i},\left\vert \psi_{i}\right\rangle \right\}  }\sum
_{i}p_{i}\left[  \sqrt{\frac{1}{\left\vert R\right\vert -\left\vert
R^{\gamma_{i}}\right\vert }}\sum_{\eta_{1},\eta_{2}\in R}\left(  \left\vert
c_{\eta_{1}}^{i}c_{\eta_{2}}^{i}\right\vert -\left\vert c_{\eta_{1}%
^{\gamma_{i}}}^{i}c_{\eta_{2}^{\gamma_{i}}}^{i}\right\vert \right)  \right],
\end{align}
}} where $\gamma_{i}$ is the partition in which the linear entropy
$S_{L}\left(  \left\vert \psi_{i}\right\rangle \left\langle \psi
_{i}\right\vert _{\gamma}\right)  $ of $\left\vert \psi_{i}\right\rangle
\left\langle \psi_{i}\right\vert $ has its minimum. By defining the
normalization factor $N_{R}:=\min_{\gamma}\left\vert R^{\gamma}\right\vert$ ,
which is the minimal value of $\left\vert R^{\gamma}\right\vert $ over all
possible bipartitions $\left\{  \gamma|\bar{\gamma}\right\}  $, we can extract
the prefactor from the convex roof summation.{\footnotesize {
\begin{align}
&  E_{m}\geq\nonumber\\
&  2\sqrt{\frac{1}{\left\vert R\right\vert -N_{R}}}\inf_{\left\{
p_{i},\left\vert \psi_{i}\right\rangle \right\}  }\sum_{i}p_{i}\left[
\sum_{\eta_{1},\eta_{2}\in R}\left(  \left\vert c_{\eta_{1}}^{i}c_{\eta_{2}%
}^{i}\right\vert -\left\vert c_{\eta_{1}^{\gamma_{i}}}^{i}c_{\eta_{2}%
^{\gamma_{i}}}^{i}\right\vert \right)  \right].
\label{eq. proof of general lower bound on GME measure}%
\end{align}
}}

The most difficult part of detecting entanglement of mixed states is a result of the
mixing of the decomposition coefficients $c_{\eta_{1}}^{i}c_{\eta_{2}}^{i}%
$. In the lab we have only the information about the mixed density matrix element
$\rho_{\eta_{1}\eta_{2}}$ but not $c_{\eta_{1}}^{i}c_{\eta_{2}}^{i}$,
therefore we must exchange the two summations in
eq.(\ref{eq. proof of general lower bound on GME measure}), and mix the
coefficients $c_{\eta_{1}}^{i}c_{\eta_{2}}^{i}$ into density matrix elements.
Therefore we estimate the summands with a bound, which is independent of the
specific partition $\gamma_{i}|\bar{\gamma}_{i}$, by adding a
summation of non-positive terms $\sum_{R^{\gamma_{i}}}\left[  \left\vert
c_{\eta_{1}}^{i}c_{\eta_{2}}^{i}\right\vert -\frac{1}{2}\left(  \left\vert
c_{\eta_{1}}^{i}\right\vert ^{2}+\left\vert c_{\eta_{2}}^{i}\right\vert
^{2}\right)  \right]  $ into the summands. {\footnotesize {
\begin{align}
&  \sum_{\left(  \eta_{1},\eta_{2}\right)  \in R\backslash R^{\gamma_{i}}%
}\left(  \left\vert c_{\eta_{1}}^{i}c_{\eta_{2}}^{i}\right\vert -\left\vert
c_{\eta_{1}^{\gamma_{i}}}^{i}c_{\eta_{2}^{\gamma_{i}}}^{i}\right\vert \right)
\nonumber\\
\geq & \sum_{\left(  \eta_{1},\eta_{2}\right)  \in R}\left(  \left\vert
c_{\eta_{1}}^{i}c_{\eta_{2}}^{i}\right\vert -\sum_{\gamma\in\Gamma\left(
\eta_{1},\eta_{2}\right)  }\left\vert c_{\eta_{1}^{\gamma}}^{i}c_{\eta
_{2}^{\gamma}}^{i}\right\vert \right) \nonumber\\
&-\frac{1}{2}\sum_{R^{\gamma_{i}}}
\left(  \left\vert c_{\eta_{1}}^{i}\right\vert ^{2}+\left\vert c_{\eta_{2}}^{i}\right\vert ^{2}\right)  \nonumber\\
\geq & \sum_{\left(  \eta_{1},\eta_{2}\right)  \in R}\left(  \left\vert
c_{\eta_{1}}^{i}c_{\eta_{2}}^{i}\right\vert -\sum_{\gamma\in\Gamma\left(
\eta_{1},\eta_{2}\right)  }\left\vert c_{\eta_{1}^{\gamma}}^{i}c_{\eta
_{2}^{\gamma}}^{i}\right\vert \right)  -\frac{1}{2}\sum_{\eta\in I\left(
R\right)  }n_{\eta}^{\gamma_{i}}\left\vert c_{\eta}^{i}\right\vert^{2},
\label{eq. proof. lower bound on GME measure}%
\end{align}
}} where $I\left(  R\right)  :=\left\{  \eta\in\mathbb{N}_{d}^{\otimes
n}:\exists\left(  \eta,\eta^{\prime}\right)  \text{ or }\left(  \eta^{\prime
},\eta\right)  \in R\right\}  $ is the set of indices contained in the set $R$,
$\Gamma\left(  \eta_{1},\eta_{2}\right)  =\left\{  \gamma|P\left(  \eta
_{1},\eta_{2}\right)  \not \in R\right\}  $ and $n_{\eta}^{\gamma_{i}}$ is the
number of vector pairs in $R^{\gamma_{i}}$ containing index $\eta$. In order
to eliminate the dependence of the partition $\gamma^{i}$, we define the maximal
value of $n_{\eta}^{\gamma}$ over all possible partitions $\left\{  \gamma
|\bar{\gamma}\right\}  $ as $N_{\eta}:=\max_{\gamma}n_{\eta}^{\gamma}$.
Then one can estimate the GME measure with
eq.(\ref{eq. proof of general lower bound on GME measure} and
\ref{eq. proof. lower bound on GME measure}) as{\footnotesize {%
\begin{align}
&  E_{m}\left(  \rho\right)  \nonumber\\
\geq &  2\sqrt{\frac{1}{\left\vert R\right\vert -N_{R}}}\sum_{\eta_{1}%
,\eta_{2}\in R}\left[  \inf_{\left\{  p_{i},\psi_{i}\right\}  }\sum_{i}%
p_{i}\left(  \left\vert c_{\eta_{1}}^{i}c_{\eta_{2}}^{i}\right\vert
-\sum_{\gamma\in\Gamma\left(  \eta_{1},\eta_{2}\right)  }\left\vert
c_{\eta_{1}^{\gamma}}^{i}c_{\eta_{2}^{\gamma}}^{i}\right\vert \right)
\right.  \nonumber\\
&  \left.  -\frac{1}{2}\sum_{\eta\in I\left(  R\right)  }N_{\eta}\left\vert
c_{\eta}^{i}\right\vert ^{2}\right].
\label{eq. proof of general lower bound on GME measure-1}%
\end{align}
}} Now one can safely exchange the summation in
eq.(\ref{eq. proof of general lower bound on GME measure-1}) and lower bound
it with the triangle inequality (i.e. $\sum_{p_{i}}p_{i}\left\vert c_{\eta_{1}%
}^{i}c_{\eta_{2}}^{i}\right\vert \geq\left\vert \rho_{\eta_{1}\eta_{2}%
}\right\vert $) and the Cauchy-Schwarz inequality (i.e. $\sum_{p_{i}}%
p_{i}\left\vert c_{\eta_{1}^{\gamma}}^{i}c_{\eta_{2}^{\gamma}}^{i}\right\vert
\leq\sqrt{\rho_{\eta_{1}^{\gamma}\eta_{1}^{\gamma}}\rho_{\eta_{2}^{\gamma}%
\eta_{2}^{\gamma}}}$). Finally we arrive at the result

{\footnotesize {
\begin{align}
&  E_{m}\left(  \rho\right)  \nonumber\\
\geq &  2\sqrt{\frac{1}{\left\vert R\right\vert -N_{R}}}\left[  \sum_{\eta
_{1},\eta_{2}\in R_{\symbol{126}}}\left(  \left\vert \rho_{\eta_{1}\eta_{2}%
}\right\vert -\sum_{\gamma\in\Gamma\left(  \eta_{1},\eta_{2}\right)  }%
\sqrt{\rho_{\eta_{1}^{\gamma}\eta_{1}^{\gamma}}\rho_{\eta_{2}^{\gamma}\eta
_{2}^{\gamma}}}\right)  \right.  \nonumber\\
&  \left.  -\frac{1}{2}\sum_{\eta\in I\left(  R\right)  }N_{\eta}%
\rho_{\eta\eta}\right]  ,\label{eq. proof of witness constructions}%
\end{align}
}} where $\rho_{\eta_{1}\eta_{2}}:=\langle\eta_{1}|\rho|\eta_{2}\rangle$.

Above is the proof of theorem \ref{theorem. lower bound of GME measure} in the case of $N_{R}:=\min_{\gamma}\left\vert R^{\gamma}\right\vert$. For the choice of $N_{R}:=\max_{\gamma}\left\vert R^{\gamma}\right\vert$, one just needs to calculate $\max_{\gamma}\left\vert R^{\gamma}\right\vert$ at the first step, i.e. eq.(\ref{eq.: linear_entropy_lower_bound_1}), then pick up $|R|-\max_{\gamma}\left\vert R^{\gamma}\right\vert$ elements from $R\backslash R^{\gamma}$ as summation region in the second line and then repeat the whole
proof above. At the end we will attain the same expression for the lower bound on
$E_{m}$ as eq.(\ref{eq. proof of witness constructions}), but with different $N_\eta$ from the ones before $N_{R}=\min_{\gamma}\left\vert R^{\gamma}\right\vert$. $N_\eta$ in this maximum choice is greater or equal to the one derived in the minimal-case. In the four-qubit singlet example in sec.\ref{sec. four-qubit-example}, the value of $N_\eta$ is exactly the same for both choices. Therefore we choose the maximum, i.e. $N_R=2$, to get a tighter lower bound on $E_m$.

\subsection{Explicit decomposition of the bipartite witness into local observables}\label{GellMann}
The measurements needed to ascertain the relevant density matrix elements in the bipartite scenario can be performed in a basis consisting of a tensor product of the generalized Gell-Mann matrices.	We continue to provide for each of the density matrix elements above their respective coefficients. 	The density matrix elements are either off diagonal elements or diagonal elements.
	The off-diagonal elements can be obtained by expectation values of the symmetric and antisymmetric generalized Gell-Mann matrices:
	\begin{align}
	 \Lambda_{s}^{12} & = & \left(\begin{smallmatrix} 0&1&0\\ 1&0&0  \\ 0&0&0  \end{smallmatrix}\right),&
	 \Lambda_{s}^{13} & = & \left(\begin{smallmatrix} 0&0&1\\ 0&0&0  \\ 1&0&0  \end{smallmatrix}\right),&
	 \Lambda_{s}^{23} & = & \left(\begin{smallmatrix} 0&0&0\\ 0&0&1  \\ 0&1&0  \end{smallmatrix}\right),&\\
	 \Lambda_{a}^{12} & = & \left(\begin{smallmatrix} 0&-i&0\\ i&0&0 \\ 0&0&0  \end{smallmatrix}\right),&
	 \Lambda_{a}^{13} & = & \left(\begin{smallmatrix} 0&0&-i\\ 0&0&0  \\ i&0&0  \end{smallmatrix}\right),&
	 \Lambda_{a}^{23} & = & \left(\begin{smallmatrix} 0&0&0\\ 0&0&-i \\ 0&-i&0  \end{smallmatrix}\right).&
	\end{align}

	They can be written as follows:

	\begin{align}
		\Re e\left[\left\langle 00 |\rho |11 \right\rangle\right] & =
		\frac{1}{2}\left\langle\Lambda_{s}^{12} \otimes \Lambda_{s}^{12}-\Lambda_{a}^{12} \otimes \Lambda_{a}^{12}\right\rangle ,	\\			
		\Re e\left[\left\langle 00 |\rho| 22 \right\rangle\right] & =
		\frac{1}{2}\left\langle\Lambda_{s}^{13} \otimes \Lambda_{s}^{13}-\Lambda_{a}^{13} \otimes \Lambda_{a}^{13}\right\rangle ,	\\	
		\Re e\left[\left\langle 11 |\rho| 22 \right\rangle\right] & =
		\frac{1}{2}\left\langle\Lambda_{s}^{23} \otimes \Lambda_{s}^{23}-\Lambda_{a}^{23} \otimes \Lambda_{a}^{23}\right\rangle .	
	\end{align}

	We now consider the terms obtained via the diagonal generalized Gell-Mann matrices.
	$\Lambda_{d}^{0} = \left(\begin{smallmatrix} 1&0&0\\ 0&1&0  \\ 0&0&1  \end{smallmatrix}\right)$,
	$\Lambda_{d}^{1} = \left(\begin{smallmatrix} 1&0&0\\ 0&-1&0 \\ 0&0&0  \end{smallmatrix}\right)$,
	$\Lambda_{d}^{2} = \frac{1}{\sqrt{3}}\left(\begin{smallmatrix} 1&0&0\\ 0&1&0  \\ 0&0&-2 \end{smallmatrix}\right)$.
	We will expand the soughtafter terms into coefficients, utilizing the following basis:
	\begin{align}
	b=\left(\begin{matrix}
	\Lambda_{d}^{0} \otimes \Lambda_{d}^{0} \\
	\Lambda_{d}^{0} \otimes \Lambda_{d}^{1} \\
	\Lambda_{d}^{0} \otimes \Lambda_{d}^{2} \\
	\Lambda_{d}^{1} \otimes \Lambda_{d}^{0} \\
	\Lambda_{d}^{1} \otimes \Lambda_{d}^{1} \\
	\Lambda_{d}^{1} \otimes \Lambda_{d}^{2} \\
	\Lambda_{d}^{2} \otimes \Lambda_{d}^{0} \\
	\Lambda_{d}^{2} \otimes \Lambda_{d}^{1} \\
	\Lambda_{d}^{2} \otimes \Lambda_{d}^{2} \\
	\end{matrix}\right).
	\end{align}
	
	For further reference the coefficients are given as:
	%im komentar steht immer die original ausgabe aus mathematica		
		\begin{align*}
			\left\langle 01| \rho| 01 \right\rangle & = \left\langle b*\left(
			\frac{1}{9},\frac{1}{6},\frac{1}{6\sqrt{3}},
			-\frac{1}{6},-\frac{1}{4},-\frac{1}{8\sqrt{3}},\frac{1}{6\sqrt{3}},\frac{1}{4\sqrt{3}},\frac{1}{12}
			\right)\right\rangle,\\			
			%{{c1 -> 1/9, c2 -> 1/6, c3 -> 1/(6 Sqrt[3]), c4 -> -(1/6),
 			% c5 -> -(1/4), c6 -> -(1/(8 Sqrt[3])), c7 -> 1/(6 Sqrt[3]),
  		% c8 -> 1/(4 Sqrt[3]), c9 -> 1/12}}
			\left\langle 10 |\rho |10 \right\rangle & = \left\langle b*\left(
			 \frac{1}{9},-\frac{1}{6},\frac{1}{6\sqrt{3}},\frac{1}{6},-\frac{1}{4},\frac{1}{8\sqrt{3}},\frac{1}{6\sqrt{3}},-\frac{1}{4\sqrt{3}},\frac{1}{12}
		 	\right)\right\rangle,\\
			%{{c1 -> 1/9, c2 -> -(1/6), c3 -> 1/(6 Sqrt[3]), c4 -> 1/6,
  		%c5 -> -(1/4), c6 -> 1/(8 Sqrt[3]), c7 -> 1/(6 Sqrt[3]),
  		%c8 -> -(1/(4 Sqrt[3])), c9 -> 1/12}}
			\left\langle 02 |\rho |02 \right\rangle & = \left\langle b*\left(
			\frac{1}{9},\frac{1}{6},\frac{1}{6\sqrt{3}},0,0,0,-\frac{1}{3\sqrt{3}},-\frac{1}{2\sqrt{3}},-\frac{1}{6}
			\right)\right\rangle,\\
			%{{c1 -> 1/9, c2 -> 1/6, c3 -> 1/(6 Sqrt[3]), c4 -> 0, c5 -> 0,
  		%c6 -> 0, c7 -> -(1/(3 Sqrt[3])), c8 -> -(1/(2 Sqrt[3])),
  		%c9 -> -(1/6)}}
			\left\langle 20| \rho| 20 \right\rangle & =
			\left\langle b*\left(
			\frac{1}{9},0,-\frac{1}{3\sqrt{3}},\frac{1}{6},0,-\frac{1}{4\sqrt{3}},\frac{1}{6\sqrt{3}},0,-\frac{1}{6}
			\right)\right\rangle,\\
			%{{c1 -> 1/9, c2 -> 0, c3 -> -(1/(3 Sqrt[3])), c4 -> 1/6, c5 -> 0,
			%  c6 -> -(1/(4 Sqrt[3])), c7 -> 1/(6 Sqrt[3]), c8 -> 0, c9 -> -(1/6)}}			
			\left\langle 12 |\rho |12 \right\rangle & = \left\langle b*\left(
			\frac{1}{9},-\frac{1}{6},\frac{1}{6\sqrt{3}},0,0,0,-\frac{1}{3\sqrt{3}},\frac{1}{2\sqrt{3}},-\frac{1}{6}
			\right)\right\rangle,\\
			%{{c1 -> 1/9, c2 -> -(1/6), c3 -> 1/(6 Sqrt[3]), c4 -> 0, c5 -> 0,
  		%c6 -> 0, c7 -> -(1/(3 Sqrt[3])), c8 -> 1/(2 Sqrt[3]), c9 -> -(1/6)}}
			\left\langle 21| \rho |21 \right\rangle & = \left\langle b*\left(
			\frac{1}{9},0,-\frac{3}{3\sqrt{3}},-\frac{1}{6},0,\frac{1}{4\sqrt{3}},\frac{1}{6\sqrt{3}},0,-\frac{1}{6}			 
			%{{c1 -> 1/9, c2 -> 0, c3 -> -(1/(3 Sqrt[3])), c4 -> -(1/6), c5 -> 0,
 			% c6 -> 1/(4 Sqrt[3]), c7 -> 1/(6 Sqrt[3]), c8 -> 0, c9 -> -(1/6)}}			
			\right)\right\rangle.
		\end{align*}

\subsection{Explicit form of the GME witness $Q_m^{(d)}$}\label{DickeWitness}
Here we recall the explicit form of the nonlinear witness from Ref.~\cite{shgh1}. Using the notation for Dicke states introduced in section \ref{Dicke} we arrive at the following lower bound
\begin{widetext}
\begin{equation}
Q_{m}^{\left(  d\right)  }=\frac{1}{m}\left[  \sum_{l,l^{\prime}=0}^{d-2}%
\sum_{\sigma}\left(  \left\vert \left\langle \alpha^{l}\left\vert
\rho\right\vert \beta^{l^{\prime}}\right\rangle \right\vert -\sum_{\delta
\in\Delta}\sqrt{\left\langle \alpha^{l}\right\vert \otimes\left\langle
\beta^{l^{\prime}}\right\vert P_{\delta}^{\dagger}\rho^{\otimes2}P_{\delta
}\left\vert \alpha^{l}\right\rangle \otimes\left\vert \beta^{l^{\prime}%
}\right\rangle }\right)  -N_{D}\sum_{l=0}^{d-2}\sum_{\alpha}\left\langle
\alpha^{l}\left\vert \rho\right\vert \alpha^{l}\right\rangle \right],
\label{eq.: GME-witness formular}%
\end{equation}
with%
\begin{align}
m &  \in\left\{  1,\cdots,\left\lfloor n/2\right\rfloor \right\}
,N_{D}=\left(  d-1\right)  m\left(  n-m-1\right)  \nonumber,\\
\sigma &  :=\left\{  \left(  \alpha,\beta\right)  :\left\vert \alpha\cap
\beta\right\vert =m-1\right\}  \nonumber,\\
\Delta &  :=\left\{
\begin{array}
[c]{cc}%
\alpha & ,l'=l\\
\left\{  \delta|\delta\subset\overline{\alpha\backslash\beta}\right\}   &
,l'<l\\
\left\{  \delta|\delta\subset\overline{\beta\backslash\alpha}\right\}   & ,l'>l
\end{array}
\right.  .\label{def. GME witness - permutation set}%
\end{align}
\end{widetext}
The properties of this witness are discussed in the main text.

\end{document}